\documentclass[a4paper,11pt]{article}

\usepackage{amsfonts,amssymb,amsmath,amsthm,dsfont,xfrac,xspace}
\usepackage{fullpage}
\usepackage{graphicx}
\usepackage{subfig}
\usepackage{cite}
\usepackage{hyperref}
\usepackage{algorithmic}
\graphicspath{{./},{fig/}}
\makeatletter
\let\NAT@parse\undefined
\makeatother
\usepackage[sort&compress, numbers]{natbib}
\usepackage{hyperref}

\newsavebox{\ieeealgbox}
\newenvironment{boxedalgorithmic}
  {\begin{lrbox}{\ieeealgbox}
   \begin{minipage}{\dimexpr\columnwidth-2\fboxsep-2\fboxrule}
   \begin{algorithmic}[1]}
  {\end{algorithmic}
   \end{minipage}
   \end{lrbox}\noindent\fbox{\usebox{\ieeealgbox}}}

\newcommand{\bsy}[1]{\boldsymbol{#1}}
\newcommand{\supp}{\text{supp}}

\newtheorem{thm}{Theorem}

\renewcommand{\thelem}{\the \numexpr (\value{thm}+1) \relax.\arabic{lem}}

\newtheorem{lemSA}{Lemma}
\newtheorem{rem}{Remark}

\newcounter{algoCounter}

\title{Improving the Correlation Lower Bound for Simultaneous Orthogonal Matching Pursuit} 

\author{ Jean-Fran\c{c}ois Determe\thanks{Jean-Fran\c{c}ois Determe and Fran\c{c}ois Horlin are with the OPERA Wireless Communications Group, Universit\'e Libre de Bruxelles, 1050 Brussels, Belgium. E-mail: jdeterme@ulb.ac.be, fhorlin@ulb.ac.be. Jean-Fran\c{c}ois Determe is funded by the Belgian National Science Foundation (F.R.S.-FNRS).} \footnotemark[2]  \quad J\'er\^{o}me Louveaux\footnotemark[2] \quad  Laurent Jacques\thanks{Laurent Jacques, J\'{e}r\^{o}me Louveaux, and Jean-Fran\c cois Determe are with the ICTEAM departement, Universit\'e catholique de Louvain. E-mail: laurent.jacques@uclouvain.be, jerome.louveaux@uclouvain.be. Laurent Jacques is funded by the Belgian National Science Foundation (F.R.S.-FNRS).} \quad Fran\c{c}ois Horlin\footnotemark[1] }

\begin{document}
\maketitle

\begin{abstract}
The simultaneous orthogonal matching pursuit (SOMP) algorithm aims to find the joint support of a set of sparse signals acquired under a multiple measurement vector model. Critically, the analysis of SOMP depends on the maximal inner product of any atom of a suitable dictionary and the current signal residual, which is formed by the subtraction of previously selected atoms. This inner product, or correlation, is a key metric to determine the best atom to pick at each iteration. This paper provides, for each iteration of SOMP, a novel lower bound of the aforementioned metric for the atoms belonging to the correct and common joint support of the multiple signals. Although the bound is obtained for the noiseless case, its main purpose is to intervene in noisy analyses of SOMP. Finally, it is shown for specific signal patterns that the proposed bound outperforms state-of-the-art results for SOMP, and orthogonal matching pursuit (OMP) as a special case. 
\end{abstract}

\section{Introduction} \label{sec:intro}

The recovery of signals possessing a sparse representation in some orthonormal basis $\bsy{\Psi}$, \textit{i.e.}, signals fully expressed  using a limited number of vectors from $\bsy{\Psi}$, acquired by means of a linear measurement process is a problem that has gained in popularity in the last decade with the emergence of the \textit{compressive sensing} (CS) \cite{donoho2006compressed, candes2006stable} theory. This paper analyzes simultaneous orthogonal matching pursuit (SOMP)  \cite{tropp2006algorithms} for this sparse signal recovery problem involving possibly more than one sparse signal to be retrieved.
\subsection{Signal model}

Let us now define our models of interest. For the sake of clarity, we assume below that $\bsy{\Psi} = \bsy{I}$ but all our results can easily be adapted to the general case. Using $\lbrack n \rbrack := \lbrace 1, \dots, n \rbrace$, we define the support of any vector $\bsy{x}$ as $\mathrm{supp} (\bsy{x}) = \lbrace j \in \lbrack n \rbrack : x_j \neq 0 \rbrace$ with $\bsy{x}$ being $s$-sparse whenever $\| \bsy{x} \|_0 := |\mathrm{supp} (\bsy{x})| \leq s$. In this context, $x_j$ is the $j$th entry of $\bsy{x}$ while $| \cdot |$ denotes the cardinality. In a single measurement vector (SMV) signal model \cite{eldar2009robust}, we consider a $|\mathcal{S}|$-sparse signal $\bsy{x}$ whose support is $\mathcal{S}$ and the corresponding measurement vector $\bsy{y} \in \mathbb{R}^{m}$ gathering measurements of $\bsy{x}$:
\begin{equation} \label{eq:defSMV}
	\bsy{y} = \bsy{\Phi} \bsy{x},
\end{equation}
where $\bsy{\Phi} \in \mathbb{R}^{m \times n}$ describes the linear measurement process being used. We find convenient to refer to the columns of the measurement matrix $\bsy{\Phi}$ as \textit{atoms}. This terminology is usually employed when dealing with dictionaries, which implicitly exist in our signal model. For $\mathcal{S} := \supp (\bsy{x})$, Equation~(\ref{eq:defSMV}) indeed rewrites $\bsy{y} = \sum_{j \in \mathcal{S}} x_j \bsy{\phi}_j$ where $\bsy{\phi}_j$ denotes the $j$th column (or atom) of $\bsy{\Phi}$. Thus, recovering $\mathcal{S}$ is equivalent to determining which set of $|\mathcal{S}|$ columns from $\bsy{\Phi}$ enables one to fully express $\bsy{y}$ using the proper linear combination.\\

Even for $m < n$, it can be shown that several algorithms of reasonable complexity can recover any sufficiently sparse signal $\bsy{x}$  provided that the matrix $\bsy{\Phi}$ satisfies some properties. Among them, the restricted isometry property (RIP) \cite{candes2006stable} is probably one of the most ubiquitous in the CS literature. A matrix $\bsy{\Phi}$ satisfies the RIP of order $s$ with restricted isometry constant (RIC) $\delta_s$ if and only if $\delta_s \in \lbrack 0, 1)$ is the smallest $\delta$ such that
\begin{equation}\label{eq:defRIPRIC}
	(1 - \delta) \| \bsy{u} \|_2^2 \leq \| \bsy{\Phi} \bsy{u} \|_2^2 \leq (1 + \delta) \| \bsy{u} \|_2^2
\end{equation}
is true for all $s$-sparse vectors $\bsy{u}$. For a given measurement matrix $\bsy{\Phi}$, the RIC determines how close the $\ell_2$-norms of any $s$-sparse signal and its associated measurement vector are. \\

The signal sensing model (\ref{eq:defSMV}) can be generalized by considering the associated multiple measurement vector (MMV) signal model \cite{eldar2009robust}
\begin{equation} \label{eq:defMMV}
	\bsy{Y} = (\bsy{y}_1, \dots, \bsy{y}_K)  = \bsy{\Phi}\; (\bsy{x}_1, \dots,  \bsy{x}_K) = \bsy{\Phi} \bsy{X}
\end{equation}
where $\bsy{Y} \in \mathbb{R}^{m \times K}$ and $\bsy{X} \in \mathbb{R}^{n \times K}$. Note that Equation~(\ref{eq:defMMV}) also rewrites $\bsy{y}_k = \bsy{\Phi} \bsy{x}_k$ for $1 \leq k \leq K$. The notion of support is extended to the matrix $\bsy{X}$ by defining $\mathrm{supp} ( \bsy{X} ) := \cup_{k \in \lbrack K \rbrack} \mathrm{supp} (\bsy{x}_k)$. Before introducing SOMP, we present some conventions.\\

\vspace*{-1px}
\textit{Conventions: }  We consider the norms $\| \bsy{x} \|_{\infty} := \max_{j \in \lbrack n \rbrack} |x_j|$ and $\| \bsy{x} \|_p := (\sum_{j = 1}^n | x_j |^p)^{1/p}$ where $1 \leq p < \infty$ and $\bsy{x} \in \mathbb{R}^n$.  In this work, any vector is a column vector. For $\mathcal{S} \subseteq \lbrack n \rbrack$, the vector $\bsy{x}_{\mathcal{S}}$ is formed by the entries of $\bsy{x}$ whose indices belong to $\mathcal{S}$. In a likewise fashion,  $\bsy{\Phi}_{\mathcal{S}}$ is defined as the matrix formed by the columns of $\bsy{\Phi}$ indexed within $\mathcal{S}$. Similarly, $\bsy{X}^{\mathcal{S}}$ contains the rows of $\bsy{X}$ indexed by $\mathcal{S}$. The Moore-Penrose pseudoinverse, transpose, and conjugate transpose of any matrix $\bsy{\Phi}$ are denoted by $\bsy{\Phi}^{+}$, $\bsy{\Phi}^\mathrm{T}$, and $\bsy{\Phi}^{*}$, respectively. The range of $\bsy{\Phi}$ is written $\mathcal{R} (\bsy{\Phi} )$. Also, the inner product of two vectors $\bsy{x}$ and $\bsy{y}$ is equal to $\langle \bsy{x} , \bsy{y} \rangle := \bsy{x}^{\mathrm{T}} \bsy{y} = \bsy{y}^{\mathrm{T}} \bsy{x} $. It is also worth defining the matrix norms $\| \bsy{A} \|_{p \rightarrow q}  := \sup_{\| \bsy{z} \|_{p} = 1} \| \bsy{A} \bsy{z} \|_{q}$. For $\bsy{A} \in \mathbb{R}^{n \times K}$, we have $\| \bsy{A} \|_{\infty \rightarrow \infty} = \max_{j \in \lbrack n \rbrack} \sum_{k=1}^K | A_{j,k} |$ as well as $\| \bsy{A} \|_{2 \rightarrow 2} = \sqrt{\lambda_{\mathrm{max}} (\bsy{A}^{*} \bsy{A}) }$ \cite[Lemma A.5]{foucart2013mathematical} where $\lambda_{\mathrm{max}}$ denotes the maximal eigenvalue. Finally, the Frobenius norm of $\bsy{A}$ is denoted by $\| \bsy{A} \|_{\mathrm{F}}$.

\subsection{Orthogonal Matching Pursuit algorithms} \label{subsec:SOMP}

We present in this section the class of OMP algorithms for SMV and MMV models.
In the event where the sparse signals $\bsy{x}_k$ to be recovered happen to share similar if not identical supports, it is interesting to perform a joint support recovery \cite{gribonval2008atoms}, \textit{i.e.}, a single and common support $\hat{\mathcal{S}}$ is jointly estimated for all the $K$ signals $\bsy{x}_k$. SOMP \cite{tropp2006algorithms}, which is described in Algorithm \ref{alg:SOMP}, performs a joint support recovery. This  algorithm iteratively picks atoms within $\bsy{\Phi}$ to simultaneously approximate the $K$ measurement vectors $\bsy{y}_k$. SOMP reduces to orthogonal matching pursuit (OMP) \cite{pati1993orthogonal, davis1997adaptive} for $K = 1$.

\begin{figure}[!h]
	\textsc{Algorithm \refstepcounter{algoCounter}\label{alg:SOMP}\arabic{algoCounter}}:\\ 
	Simultaneous orthogonal matching pursuit (SOMP)\\
	
	\vspace{-2mm}
	\begin{boxedalgorithmic}
		\small
		\REQUIRE $\bsy{Y} \in \mathbb{R}^{m \times K}$, $\bsy{\Phi} \in \mathbb{R}^{m \times n}$, $s \geq 1$
		\STATE Initialization: $\bsy{R}^{(0)} \leftarrow \bsy{Y}$ and $\mathcal{S}_0 \leftarrow \emptyset$
		\STATE $t \leftarrow 0$
		\WHILE{$t < s$}
		\STATE Determine the atom of $\bsy{\Phi}$ to be included in the support: \\ $j_t \leftarrow \mathrm{argmax}_{j \in \lbrack n \rbrack} ( \| (\bsy{R}^{(t)})^{\mathrm{T}} \bsy{\phi}_j \|_1 )$
		\STATE Update the support : $\mathcal{S}_{t+1} \leftarrow \mathcal{S}_{t} \cup \left\lbrace j_t \right\rbrace$
		\STATE Projection of each measurement vector onto $\mathcal{R}(\boldsymbol{\Phi}_{\mathcal{S}_{t+1}})$: \\$\bsy{Y}^{(t+1)} \leftarrow \boldsymbol{\Phi}_{\mathcal{S}_{t+1}} \boldsymbol{\Phi}_{\mathcal{S}_{t+1}}^{+} \bsy{Y}$
		\STATE Projection of each measurement vector onto $\mathcal{R}(\boldsymbol{\Phi}_{\mathcal{S}_{t+1}})^{\perp}$~: \\ $\bsy{R}^{(t+1)} \leftarrow \bsy{Y} - \bsy{Y}^{(t+1)}$
		\STATE $t \leftarrow t + 1$
		\ENDWHILE
		\RETURN $\mathcal{S}_s$ \COMMENT{Support at last step}
	\end{boxedalgorithmic}
\end{figure}

At each iteration $t$, SOMP adds one atom to the estimated support (step 5). The criterion to determine which atom to include is to pick the atom maximizing $\| (\bsy{R}^{(t)})^{\mathrm{T}} \bsy{\phi}_j \|_1 = \sum_{k=1}^K | \langle \bsy{r}_k^{(t)}, \bsy{\phi}_j \rangle |$ (step 4) for the current residual matrix $\bsy{R}^{(t)}$ where $\bsy{r}_k^{(t)}$ denotes the $k$th column of $\bsy{R}^{(t)}$. Note that the previous sum is a way to simultaneously account for all the measurement vectors $\bsy{y}_k$ and their corresponding residuals $\bsy{r}_k^{(t)}$. The residual is then updated so that it is orthogonal to the subspace spanned by the atoms indexed by the current estimated support (steps 5 and 6.) The orthogonal projection matrix $\bsy{P}^{(t)} := \boldsymbol{\Phi}_{\mathcal{S}_{t}} \boldsymbol{\Phi}_{\mathcal{S}_{t}}^{+}$ allows to perform the projection onto $\mathcal{R}(\boldsymbol{\Phi}_{\mathcal{S}_{t+1}})$, \textit{i.e.}, the space spanned by the columns of $\boldsymbol{\Phi}_{\mathcal{S}_{t+1}}$. Using the $\ell_1$-norm for the decision criterion of SOMP is not the only possible choice. Generally, $p$-SOMP refers to the variant of SOMP for which the $\ell_p$-norm intervenes \cite{gribonval2008atoms}. Unless otherwise specified, we assume that SOMP uses the $\ell_1$-norm. The algorithm finishes when the size of the estimated support reaches $s$. If possible, $s$ is usually chosen close to $|\mathcal{S}|$.

\subsection{Contribution and its connection with the noisy case} \label{subsec:contribNoisyLink}
In this paper, we provide a novel RIP-based lower bound of the quantity $\| \bsy{\Phi}_{\mathcal{S}}^{\mathrm{T}} \bsy{R}^{(t)} \|_{\infty \rightarrow \infty} = \max_{j \in \mathcal{S}} ( \| (\bsy{R}^{(t)})^{\mathrm{T}} \bsy{\phi}_j \|_1 )$, which is the maximum SOMP metric among the correct atoms without noise. In particular, we are interested in a lower bound expressed as $\| \bsy{\Phi}_{\mathcal{S}}^{\mathrm{T}} \bsy{R}^{(t)} \|_{\infty \rightarrow \infty} \geq \psi \tau_X$ where $\psi$ only depends on $\mathcal{S}$ and $\bsy{\Phi}$ while $\tau_X$ is determined by $\bsy{X}$. This type of bound typically intervenes in noisy analyses \cite{cai2011orthogonal, dan2014robustness, determe2016noise} of OMP or SOMP, \textit{i.e.}, for the signal model $\bsy{Y} = \bsy{\Phi} \bsy{X} + \bsy{E}$ where $\bsy{E} = (\bsy{e}_1, \dots, \bsy{e}_K)$ is the noise term. For example, \cite[Theorem 3]{determe2016noise} shows that, if each noise vector obeys $\bsy{e}_k \sim \mathcal{N}(\bsy{0}, \sigma_k^2 \bsy{I}_{m \times m})$, then the probability of SOMP identifying at least one incorrect support entry during $s$ iterations is upper bounded by $\gamma(s, |\mathcal{S}|, n) \exp (-\Delta \mathbb{E}^2/(8\| \bsy{\sigma} \|_2^2))$ where $\gamma(s, |\mathcal{S}|, n)$ increases with $s$, $|\mathcal{S}|$, and $n$. $\Delta \mathbb{E}$ can be written as $\Delta \mathbb{E} = (1-1/\Gamma)\psi \tau_X - \sqrt{2/\pi} \| \bsy{\sigma} \|_1$ where the lower bound $\psi \tau_X$ intervenes. The quantity $\Gamma$, which is not studied in the present paper, connects the highest noiseless SOMP metrics for correct and incorrect atoms, \textit{i.e.}, $\Gamma$ lower bounds the ratio $\| \bsy{\Phi}_{\mathcal{S}}^{\mathrm{T}} \bsy{R}^{(t)} \|_{\infty \rightarrow \infty}/\| \bsy{\Phi}_{\overline{\mathcal{S}}}^{\mathrm{T}} \bsy{R}^{(t)} \|_{\infty \rightarrow \infty}$ (see \cite{determe2016exact} for more details). It is worth pointing out again that, despite their intervening in the noisy analysis, the quantities $\Gamma$, $\psi$, and $\tau_X$ are all defined on the basis of noiseless signals. 

\subsection{Outline \& related work overview}

Section~\ref{sec:contrib} explains in details the contribution. Section~\ref{sec:relWork} states and comments an alternative bound in the literature. Finally, Section~\ref{sec:discussCompare} compares our contribution against the alternative lower bound. It is shown that, under several common sensing scenarios, our result outperforms its counterpart, even for $K=1$, \textit{i.e.}, for OMP.  Our contribution can be used in several theoretical analyses \cite{cai2011orthogonal, dan2014robustness} of OMP and SOMP in the noisy case by replacing the older bound of Section~\ref{sec:relWork} (Theorem~\ref{thm:frobLBCorrelationResidual}) by the one obtained in this paper, \textit{i.e.}, Theorem~\ref{thm:newLBCorrelationResidual}. For example, \cite[Lemma 4.1]{dan2014robustness}  can be  partially replaced by Theorem~\ref{thm:newLBCorrelationResidual}. The cases under which this replacement leads to less stringent conditions on whether an iteration is successful are hence discussed in Section~\ref{sec:discussCompare}. Other related algorithms include CoSaMP \cite{needell2009cosamp}, Subspace pursuit \cite{dai2009subspace}, and orthogonal matching pursuit with replacement (OMPR) \cite{jain2011orthogonal}. Since the decisions of these algorithms also rely on the highest inner products of a residual and the atoms, our methodology might provide relevant insights for them as well.

\section{Contribution}\label{sec:contrib}

Lemma~\ref{lem:technicalRes} is an upper bound on $\| \bsy{A} \|_{\infty \rightarrow \infty}$ depending on the distribution of the eigenvalues of $\bsy{A}$. This lemma is needed in the proof of our main contribution, \textit{i.e.}, Theorem~\ref{thm:newLBCorrelationResidual}.

\begin{lemSA}\label{lem:technicalRes}
	Let $\alpha \in \mathbb{R}$ and $\bsy{A} \in \mathbb{R}^{d \times d}$ be a normal matrix, \textit{i.e.}, $\bsy{A} \bsy{A}^{\mathrm{T}} = \bsy{A}^{\mathrm{T}} \bsy{A}$. Let the vector $\bsy{\theta}(\alpha)$ be composed of the elements $\theta_j(\alpha) = \lambda_j - \alpha$ where $\lambda_j$ is the $j$th eigenvalue of $\bsy{A}$. Then, $
	\| \bsy{A} \|_{\infty \rightarrow \infty} \leq |\alpha| + \sqrt{d} \| \bsy{\theta} (\alpha) \|_{\infty}$.
\end{lemSA}%
\begin{proof}
	The spectral theorem establishes that normal matrices are unitarily diagonalizable. We thus consider the eigenvalue decomposition $\bsy{A} = \bsy{Q} \bsy{\Lambda} \bsy{Q}^{*}$ where $\bsy{Q}$ is unitary. Defining $\bsy{\Delta}(\alpha) := \mathrm{diag}(\bsy{\theta}(\alpha))$, we have $\bsy{\Lambda} = \alpha \bsy{I} + \bsy{\Delta}(\alpha)$. Thus, applying the triangle inequality yields
	\begin{equation*}
		\| \bsy{A} \|_{\infty \rightarrow \infty} \leq |\alpha| \underbrace{ \|  \bsy{Q} \bsy{Q}^{*} \|_{\infty \rightarrow \infty}}_{= 1} + \| \bsy{Q} \bsy{\Delta}(\alpha) \bsy{Q}^{*}\|_{\infty \rightarrow \infty}.
	\end{equation*}
	For $\bsy{B} \in \mathbb{R}^{d \times d}$, the inequality $\| \bsy{B} \|_{\infty \rightarrow \infty} \leq \sqrt{d} \| \bsy{B} \|_{2 \rightarrow 2}$ \cite{golub2012matrix} provides $ \| \bsy{A} \|_{\infty \rightarrow \infty} \leq |\alpha|  + \sqrt{d} \| \bsy{Q} \bsy{\Delta}(\alpha) \bsy{Q}^{*}\|_{2 \rightarrow 2} = |\alpha|  + \sqrt{d} \| \bsy{\theta}(\alpha) \|_{\infty}$.
\end{proof}

Making use of the inequality derived in Lemma~\ref{lem:technicalRes}, Theorem~\ref{thm:newLBCorrelationResidual} states our novel lower bound on $\| \bsy{\Phi}_{\mathcal{S}}^{\mathrm{T}} \bsy{R}^{(t)} \|_{\infty \rightarrow \infty}$.

\begin{thm}\label{thm:newLBCorrelationResidual}
	Let us assume that SOMP has picked only correct atoms before iteration $t$, \textit{i.e.}, $\mathcal{S}_t \subset \mathcal{S}$. We denote $\mathcal{J}_t = \mathcal{S} \backslash \mathcal{S}_t$ the set that contains the indices of the correct atoms yet to be selected at iteration $t$. If $\bsy{\Phi}$ satisfies the RIP with $|\mathcal{S}|$-th RIC $\delta_{|\mathcal{S}|}  < 1$, then
	\begin{equation}
		\|\bsy{\Phi}_{\mathcal{S}}^{\mathrm{T}} \bsy{R}^{(t)} \|_{\infty \rightarrow \infty}  \geq  \dfrac{(1-\delta_{|\mathcal{S}|}) (1+\delta_{|\mathcal{S}|}) }{1 + \sqrt{|\mathcal{S}|-t} \; \delta_{|\mathcal{S}|}} \| \bsy{X}^{\mathcal{J}_t} \|_{\infty \rightarrow \infty}.
	\end{equation}
	\begin{proof} 
		We have $\bsy{\Phi}_{\mathcal{S}}^{\mathrm{T}} \bsy{R}^{(t)} =
		\bsy{\Phi}_{\mathcal{S}}^{\mathrm{T}} (\bsy{I} - \bsy{P}^{(t)}) \bsy{\Phi}_{\mathcal{S}} \bsy{X}^{\mathcal{S}} = \bsy{\Phi}_{\mathcal{S}}^{\mathrm{T}} (\bsy{I} - \bsy{P}^{(t)}) \bsy{\Phi}_{\mathcal{J}_t} \bsy{X}^{\mathcal{J}_t}$ because $(\bsy{I} - \bsy{P}^{(t)}) \bsy{\Phi}_{\mathcal{S}} \bsy{X}^{\mathcal{S}} = (\bsy{I} - \bsy{P}^{(t)}) (\bsy{\Phi}_{\mathcal{S}_t} \bsy{X}^{\mathcal{S}_t} + \bsy{\Phi}_{\mathcal{J}_t} \bsy{X}^{\mathcal{J}_t}) = (\bsy{I} - \bsy{P}^{(t)})  \bsy{\Phi}_{\mathcal{J}_t} \bsy{X}^{\mathcal{J}_t}$. For $j \in \mathcal{S}_t$, $\langle \bsy{\phi}_j,   (\bsy{I} - \bsy{P}^{(t)}) \bsy{z} \rangle = 0$ for every vector $\bsy{z}$. Thus,
		
		\begin{align*}
			\| \bsy{\Phi}_{\mathcal{S}}^{\mathrm{T}} \bsy{R}^{(t)} \|_{\infty \rightarrow \infty}  & =  \| \bsy{\Phi}_{\mathcal{S}}^{\mathrm{T}} (\bsy{I} - \bsy{P}^{(t)}) \bsy{\Phi}_{\mathcal{J}_t} \bsy{X}^{\mathcal{J}_t}\|_{\infty \rightarrow \infty}\\
			& = \max_{j \in \mathcal{S}} \sum_{k=1}^K | \langle \bsy{\phi}_j,   (\bsy{I} - \bsy{P}^{(t)}) (\bsy{\Phi}_{\mathcal{J}_t} \bsy{X}^{\mathcal{J}_t})_k \rangle |\\
			& = \max_{j \in \mathcal{J}_t} \sum_{k=1}^K | \langle \bsy{\phi}_j,   (\bsy{I} - \bsy{P}^{(t)}) (\bsy{\Phi}_{\mathcal{J}_t} \bsy{X}^{\mathcal{J}_t})_k \rangle |\\
			& = \| \bsy{\Phi}_{\mathcal{J}_t}^{\mathrm{T}} (\bsy{I} - \bsy{P}^{(t)}) \bsy{\Phi}_{\mathcal{J}_t} \bsy{X}^{\mathcal{J}_t}\|_{\infty \rightarrow \infty}.
		\end{align*}
		We follow the steps of \cite[Proof of Theorem 10]{gribonval2008atoms} and use the inequality $\| \bsy{A} \bsy{B} \|_{\infty \rightarrow \infty} \leq \| \bsy{A} \|_{\infty \rightarrow \infty} \| \bsy{B} \|_{\infty \rightarrow \infty}$. If $\bsy{B} = \bsy{C} \bsy{D}$ and $\bsy{C}$ is invertible, then, with $\bsy{A} = \bsy{C}^{-1}$, the inequality above implies $\| \bsy{C}^{-1} \bsy{C} \bsy{D} \|_{\infty \rightarrow \infty} = \| \bsy{D} \|_{\infty \rightarrow \infty} \leq \| \bsy{C}^{-1} \|_{\infty \rightarrow \infty} \| \bsy{C} \bsy{D} \|_{\infty \rightarrow \infty}$. Replacing $\bsy{D}$ with $\bsy{X}^{\mathcal{J}_t}$ and $\bsy{C}$ with $\bsy{\Phi}_{\mathcal{J}_t}^{\mathrm{T}} (\bsy{I} - \bsy{P}^{(t)}) \bsy{\Phi}_{\mathcal{J}_t}$ yields
		\begin{equation*}
			\| \bsy{\Phi}_{\mathcal{S}}^{\mathrm{T}} \bsy{R}^{(t)}  \|_{\infty \rightarrow \infty} \geq  \dfrac{\| \bsy{X}^{\mathcal{J}_t} \|_{\infty \rightarrow \infty}}{\| (\bsy{\Phi}_{\mathcal{J}_t}^{\mathrm{T}} (\bsy{I} - \bsy{P}^{(t)}) \bsy{\Phi}_{\mathcal{J}_t})^{-1} \|_{\infty \rightarrow \infty}}
		\end{equation*}
		It can be shown \cite[Lemma 5]{cai2011orthogonal} that, under the condition $\mathcal{S}_t \subseteq \mathcal{S}$, we have  $\lambda_{\mathrm{min}} (\bsy{\Phi}_{\mathcal{J}_t}^{\mathrm{T}} (\bsy{I} - \bsy{P}^{(t)}) \bsy{\Phi}_{\mathcal{J}_t}) \geq \lambda_{\mathrm{min}} (\bsy{\Phi}_{\mathcal{S}}^{\mathrm{T}} \bsy{\Phi}_{\mathcal{S}}) \geq 1 - \delta_{|\mathcal{S}|} > 0$ (see \cite[Remark 1]{dai2009subspace} regarding the penultimate inequality). The absence of zero eigenvalues thus shows that the matrix $\bsy{\Phi}_{\mathcal{J}_t}^{\mathrm{T}} (\bsy{I} - \bsy{P}^{(t)}) \bsy{\Phi}_{\mathcal{J}_t}$ is full rank and invertible. Note that it can also be shown \cite[Lemma 5]{cai2011orthogonal} that  $\lambda_{\mathrm{max}} (\bsy{\Phi}_{\mathcal{J}_t}^{\mathrm{T}} (\bsy{I} - \bsy{P}^{(t)}) \bsy{\Phi}_{\mathcal{J}_t}) \leq  \lambda_{\mathrm{max}} (\bsy{\Phi}_{\mathcal{S}}^{\mathrm{T}} \bsy{\Phi}_{\mathcal{S}}) \leq 1 + \delta_{|\mathcal{S}|}$  (see \cite[Remark 1]{dai2009subspace} for the last inequality). As a result, the eigenvalues of $(\bsy{\Phi}_{\mathcal{J}_t}^{\mathrm{T}} (\bsy{I} - \bsy{P}^{(t)}) \bsy{\Phi}_{\mathcal{J}_t})^{-1}$ belong to $\lbrack 1/(1+\delta_{|\mathcal{S}|}); 1/(1-\delta_{|\mathcal{S}|}) \rbrack$. Due to the idempotency and symmetry of the orthogonal projectors, $\bsy{\Phi}_{\mathcal{J}_t}^{\mathrm{T}} (\bsy{I} - \bsy{P}^{(t)}) \bsy{\Phi}_{\mathcal{J}_t} = \bsy{\Phi}_{\mathcal{J}_t}^{\mathrm{T}} (\bsy{I} - \bsy{P}^{(t)})^{\mathrm{T}} (\bsy{I} - \bsy{P}^{(t)}) \bsy{\Phi}_{\mathcal{J}_t}$, which implies that the matrix of interest and its inverse are symmetric and consequently normal.\\
		
		\noindent Instead of using the coherence of $\bsy{\Phi}$ as in \cite{gribonval2008atoms}, we consider the eigenvalues of  $(\bsy{\Phi}_{\mathcal{J}_t}^{\mathrm{T}} (\bsy{I} - \bsy{P}^{(t)}) \bsy{\Phi}_{\mathcal{J}_t})^{-1}$. 
		Since this last matrix is normal, Lemma~\ref{lem:technicalRes} can be used with 
		\begin{equation*}
			\alpha = \dfrac{1}{2} \left( \dfrac{1}{1+\delta_{|\mathcal{S}|}} + \dfrac{1}{1-\delta_{|\mathcal{S}|}} \right) = \dfrac{1}{(1+\delta_{|\mathcal{S}|})(1-\delta_{|\mathcal{S}|})},
		\end{equation*}
		which is the arithmetic mean of the lowest and highest possible eigenvalues of $(\bsy{\Phi}_{\mathcal{J}_t}^{\mathrm{T}} (\bsy{I} - \bsy{P}^{(t)}) \bsy{\Phi}_{\mathcal{J}_t})^{-1}$ given the RIP with the RIC $\delta_{|\mathcal{S}|}$. For such a choice,
		\begin{equation*}
			\| \bsy{\theta} (\alpha) \|_{\infty} \leq \dfrac{1}{1-\delta_{|\mathcal{S}|}} - \alpha = \dfrac{\delta_{|\mathcal{S}|}}{(1+\delta_{|\mathcal{S}|})(1-\delta_{|\mathcal{S}|})}.
		\end{equation*}
		Thus, Lemma~\ref{lem:technicalRes} yields
		\begin{equation*}
			\| (\bsy{\Phi}_{\mathcal{J}_t}^{\mathrm{T}} (\bsy{I} - \bsy{P}^{(t)}) \bsy{\Phi}_{\mathcal{J}_t})^{-1} \|_{\infty \rightarrow \infty} \leq  \dfrac{1 + \sqrt{|\mathcal{J}_t|}\delta_{|\mathcal{S}|}}{(1+\delta_{|\mathcal{S}|})(1-\delta_{|\mathcal{S}|})},
		\end{equation*}
		which concludes the proof.
	\end{proof}
\end{thm}

Theorem~\ref{thm:newLBCorrelationResidual} essentially states that the maximal correlation obtained among the correct atoms, \textit{i.e.}, $\|\bsy{\Phi}_{\mathcal{S}}^{\mathrm{T}} \bsy{R}^{(t)} \|_{\infty \rightarrow \infty}$, is lower bounded by a quantity proportional to $\| \bsy{X}^{\mathcal{J}_t} \|_{\infty \rightarrow \infty} = \max_{j \in \mathcal{J}_t} \| \bsy{X}^{\lbrace j \rbrace} \|_1 = \max_{j \in \mathcal{J}_t} \sum_{k=1}^K |X_{j,k}|$. As already stated in  Section~\ref{subsec:contribNoisyLink}, lower bounds on $\|\bsy{\Phi}_{\mathcal{S}}^{\mathrm{T}} \bsy{R}^{(t)} \|_{\infty \rightarrow \infty}$ in the noiseless case play a role when determining the performance of SOMP when additive noise is included in the signal model. The properties (including the sharpness) of Theorem~\ref{thm:newLBCorrelationResidual} will be discussed in Section~\ref{sec:discussCompare}

\section{Related work}\label{sec:relWork}

Let us now compare our lower bound on $	\|\bsy{\Phi}_{\mathcal{S}}^{\mathrm{T}} \bsy{R}^{(t)} \|_{\infty \rightarrow \infty}$, \textit{i.e.}, Theorem~\ref{thm:newLBCorrelationResidual}, to another important one in the literature, \textit{i.e.}, Theorem~\ref{thm:frobLBCorrelationResidual}. To the best of the authors' knowledge, Theorem~\ref{thm:frobLBCorrelationResidual} was first obtained in \cite[Section 3.1]{wang2013performance} for $2$-SOMP. In the SMV case, the inequality $	\|\bsy{\Phi}_{\mathcal{S}}^{\mathrm{T}} \bsy{r}^{(t)} \|_{\infty}  \geq  \frac{\lambda_{\mathrm{min}} ( \bsy{\Phi}_{\mathcal{S}}^{\mathrm{T}} \bsy{\Phi}_{\mathcal{S}})}{\sqrt{|\mathcal{S}| - t}} \left\| \bsy{x}_{\mathcal{J}_t} \right\|_{\mathrm{2}}$ was first obtained in \cite[Section V]{cai2011orthogonal} for OMP, which immediately yields Theorem~\ref{thm:frobLBCorrelationResidual} for $K=1$ when using the inequality $\lambda_{\mathrm{min}} ( \bsy{\Phi}_{\mathcal{S}}^{\mathrm{T}} \bsy{\Phi}_{\mathcal{S}}) \geq 1-\delta_{|\mathcal{S}|}$ \cite[Remark 1]{dai2009subspace}.

\begin{thm}\label{thm:frobLBCorrelationResidual}
	Let us assume that SOMP has picked only correct atoms before iteration $t$, \textit{i.e.}, $\mathcal{S}_t \subset \mathcal{S}$, with $\mathcal{J}_t = \mathcal{S} \backslash \mathcal{S}_t$ containing the indices of the correct atoms yet to be selected at iteration $t$. If $\bsy{\Phi}$ satisfies the RIP with $|\mathcal{S}|$-th RIC $\delta_{|\mathcal{S}|}  < 1$, then
	\begin{equation}
		\|\bsy{\Phi}_{\mathcal{S}}^{\mathrm{T}} \bsy{R}^{(t)} \|_{\infty \rightarrow \infty}  \geq (1-\delta_{|\mathcal{S}|}) \dfrac{1}{\sqrt{|\mathcal{S}| - t}} \left\| \bsy{X}^{\mathcal{J}_t} \right\|_{\mathrm{F}}.
	\end{equation}
\end{thm}
\begin{proof}
	Rearranging the results in \cite[Section 3.1]{wang2013performance} shows that Theorem~\ref{thm:frobLBCorrelationResidual} is true if $\max_{j \in \mathcal{S}}  \| (\bsy{R}^{(t)})^{\mathrm{T}} \bsy{\phi}_j \|_1 $ is replaced by $\max_{j \in \mathcal{S}}  \| (\bsy{R}^{(t)})^{\mathrm{T}} \bsy{\phi}_j \|_2 $, \textit{i.e.}, if $2$-SOMP is used instead of $1$-SOMP. Since $\| \bsy{x} \|_1 \geq \| \bsy{x} \|_2$ for all $\bsy{x}$,\\ $\|\bsy{\Phi}_{\mathcal{S}}^{\mathrm{T}} \bsy{R}^{(t)} \|_{\infty \rightarrow \infty} = \max_{j \in \mathcal{S}}  \| (\bsy{R}^{(t)})^{\mathrm{T}} \bsy{\phi}_j \|_1 \geq  \max_{j \in \mathcal{S}}  \| (\bsy{R}^{(t)})^{\mathrm{T}} \bsy{\phi}_j \|_2 \geq  (1-\delta_{|\mathcal{S}|})
	\| \bsy{X}^{\mathcal{J}_t} \|_{\mathrm{F}}/ \sqrt{|\mathcal{S}| - t} $, which concludes the proof.
\end{proof}
The quantity $\| \bsy{X}^{\mathcal{J}_t} \|_{\mathrm{F}}^2$ rewrites $\sum_{j \in \mathcal{J}_t} \| \bsy{X}^{\lbrace j \rbrace} \|_2^2$, which is the sum of the squared $\ell_2$-norms of each row of $\bsy{X}$ indexed by $\mathcal{J}_t$. Since each row of $\bsy{X}$ can be interpreted as the coefficient vector associated with one particular atom, $\| \bsy{X}^{\mathcal{J}_t} \|_{\mathrm{F}}^2/(|\mathcal{S}| - t)$ is the average energy of the coefficients associated with the atoms indexed by $\mathcal{J}_t$.

\section{Comparison with related works \& Discussions}\label{sec:discussCompare}

The rest of this section is dedicated to the comparison of Theorem~\ref{thm:newLBCorrelationResidual}, \textit{i.e.}, our contribution, with Theorem~\ref{thm:frobLBCorrelationResidual}. To determine which bound is the better at iteration $t$, we introduce the following quantity
\begin{equation}\label{eq:RatioFrobNewLB}
	r(|\mathcal{S}|, \mathcal{J}_t) := \dfrac{\sqrt{|\mathcal{J}_t|} (1+\delta_{|\mathcal{S}|})}{1 + \sqrt{|\mathcal{J}_t|} \; \delta_{|\mathcal{S}|} } \dfrac{\| \bsy{X}^{\mathcal{J}_t} \|_{\infty \rightarrow \infty}}{\| \bsy{X}^{\mathcal{J}_t} \|_{\mathrm{F}}},
\end{equation}
which is the ratio of the lower bound of Theorem~\ref{thm:newLBCorrelationResidual} to that of Theorem~\ref{thm:frobLBCorrelationResidual}. Our contribution thereby improves the analysis of SOMP when $r(|\mathcal{S}|, \mathcal{J}_t)  > 1$. For $\bsy{B} \in \mathbb{R}^{|\mathcal{J}_t| \times K}$, we have $\| \bsy{B} \|_{\infty \rightarrow \infty} \leq \sqrt{K} \| \bsy{B} \|_{2 \rightarrow 2} \leq \sqrt{K} \| \bsy{B} \|_{\mathrm{F}}$ \cite{golub2012matrix}. For $b_j := \| \bsy{B}^{\lbrace j \rbrace}\|_2$ (where $\bsy{b} \in \mathbb{R}^{|\mathcal{J}_t|}$), we obtain  $\| \bsy{B} \|_{\infty \rightarrow \infty} = \max_{j \in \lbrack |\mathcal{J}_t| \rbrack} \| \bsy{B}^{\lbrace j \rbrace} \|_1 \geq \max_{j \in \lbrack |\mathcal{J}_t| \rbrack} \| \bsy{B}^{\lbrace j \rbrace} \|_2 = \| \bsy{b} \|_{\infty} \geq (1/\sqrt{|\mathcal{J}_t|}) \| \bsy{b} \|_{2} = (1/\sqrt{|\mathcal{J}_t|}) \| \bsy{B} \|_{\mathrm{F}}$. As a result,
\begin{align*}
	\dfrac{1+\delta_{|\mathcal{S}|}}{1+\sqrt{|\mathcal{J}_t|}\delta_{|\mathcal{S}|}}  \leq r(|\mathcal{S}|, \mathcal{J}_t) & \leq \dfrac{(1+\delta_{|\mathcal{S}|})\sqrt{|\mathcal{J}_t|} \sqrt{K}}{1+\sqrt{|\mathcal{J}_t|}\delta_{|\mathcal{S}|}}\\
	& \leq  \sqrt{K}\dfrac{1+\delta_{|\mathcal{S}|}}{\delta_{|\mathcal{S}|}}.
\end{align*}
The proposed analysis of $r(|\mathcal{S}|, \mathcal{J}_t)$ is realized for four different reconstruction scenarios that are discussed hereafter.
\subsection{Case 1: A single dominant row within $\bsy{X}$}
We assume that $\| \bsy{X}^{\mathcal{J}_t} \|_{\infty \rightarrow \infty} \simeq \| \bsy{X}^{\lbrace j_d \rbrace} \|_1$ and 
$\| \bsy{X}^{\mathcal{J}_t} \|_{\mathrm{F}} \simeq \| \bsy{X}^{\lbrace j_d \rbrace} \|_2$ for some $j_d \in \mathcal{S}$. This situation occurs whenever the entries of the $j_d$th row of $\bsy{X}$ have magnitudes overwhelmingly higher than those of all the other rows combined. Then, 
\begin{equation}\label{eq:case1}
	r(|\mathcal{S}|, \mathcal{J}_t) \simeq \dfrac{\sqrt{|\mathcal{J}_t|} (1+\delta_{|\mathcal{S}|})}{1 + \sqrt{|\mathcal{J}_t|} \; \delta_{|\mathcal{S}|} } \dfrac{\| \bsy{X}^{\lbrace j_d \rbrace} \|_1}{\| \bsy{X}^{\lbrace j_d \rbrace} \|_2} \geq \dfrac{\sqrt{|\mathcal{J}_t|} (1+\delta_{|\mathcal{S}|})}{1 + \sqrt{|\mathcal{J}_t|} \; \delta_{|\mathcal{S}|} }
\end{equation}
where $1 \leq \| \bsy{x} \|_1/ \| \bsy{x} \|_2 \leq \sqrt{K}$ for all $\bsy{x} \in \mathbb{R}^K$. Theorem~\ref{thm:newLBCorrelationResidual} always outperforms Theorem~\ref{thm:frobLBCorrelationResidual} in this case since the RHS of Equation~(\ref{eq:case1}) is higher than $1$. Interestingly, it remains true in the SMV setting, thereby making our contribution superior to the state-of-the-art result of \cite{cai2011orthogonal} for OMP. If $\bsy{X}^{\lbrace j_d \rbrace}$ is $1$-sparse, then $r(|\mathcal{S}|, \mathcal{J}_t) = \sqrt{|\mathcal{J}_t|} (1+\delta_{|\mathcal{S}|})/(1 + \sqrt{|\mathcal{J}_t|} \; \delta_{|\mathcal{S}|})$ while $r(|\mathcal{S}|, \mathcal{J}_t) = \sqrt{K} \sqrt{|\mathcal{J}_t|} (1+\delta_{|\mathcal{S}|})/(1 + \sqrt{|\mathcal{J}_t|} \; \delta_{|\mathcal{S}|})$ whenever the entries of $\bsy{X}^{\lbrace j_d \rbrace}$ have identical absolute values. These first observations suggest that the improvements resulting from using Theorem~\ref{thm:newLBCorrelationResidual} instead of Theorem~\ref{thm:frobLBCorrelationResidual} increase with $K$.
\subsection{Case 2: Identical magnitudes}

We assume that $|X_{j,k}| \simeq \mu_X > 0$ for each $(j,k) \in \mathcal{J}_t \times \lbrack K \rbrack$. Thus, we have $\| \bsy{X}^{\mathcal{J}_t} \|_{\infty \rightarrow \infty} \simeq K \mu_X$ and $\| \bsy{X}^{\mathcal{J}_t} \|_{\mathrm{F}} \simeq \sqrt{K} \sqrt{|\mathcal{J}_t|} \mu_X$. As a result, we obtain
\begin{equation}\label{eq:case2}
	r(|\mathcal{S}|, \mathcal{J}_t) \simeq   \dfrac{\sqrt{K} (1+\delta_{|\mathcal{S}|})}{1 + \sqrt{|\mathcal{J}_t|} \; \delta_{|\mathcal{S}|} }.
\end{equation}
As depicted in Figure~\ref{fig:case2LBComparison}, the situation might be favorable to both lower bounds depending on the value of $K$, $|\mathcal{J}_t|$, and $\delta_{|S|}$. In this case, Theorem~\ref{thm:newLBCorrelationResidual} is always worse than its counterpart if $K = 1$, $|\mathcal{J}_t| > 1$, and $\delta_{|\mathcal{S}|} > 0$. As a general rule, our contribution tends to get better than Theorem~\ref{thm:frobLBCorrelationResidual} whenever $K$ increases, the RIC $\delta_{|\mathcal{S}|}$ approaches $0$, or the number of correct atoms yet to be recovered, \textit{i.e.}, $|\mathcal{J}_t|$, tends to $1$.

\begin{figure}[!h]%
	\centering
	\includegraphics[width=12.5cm]{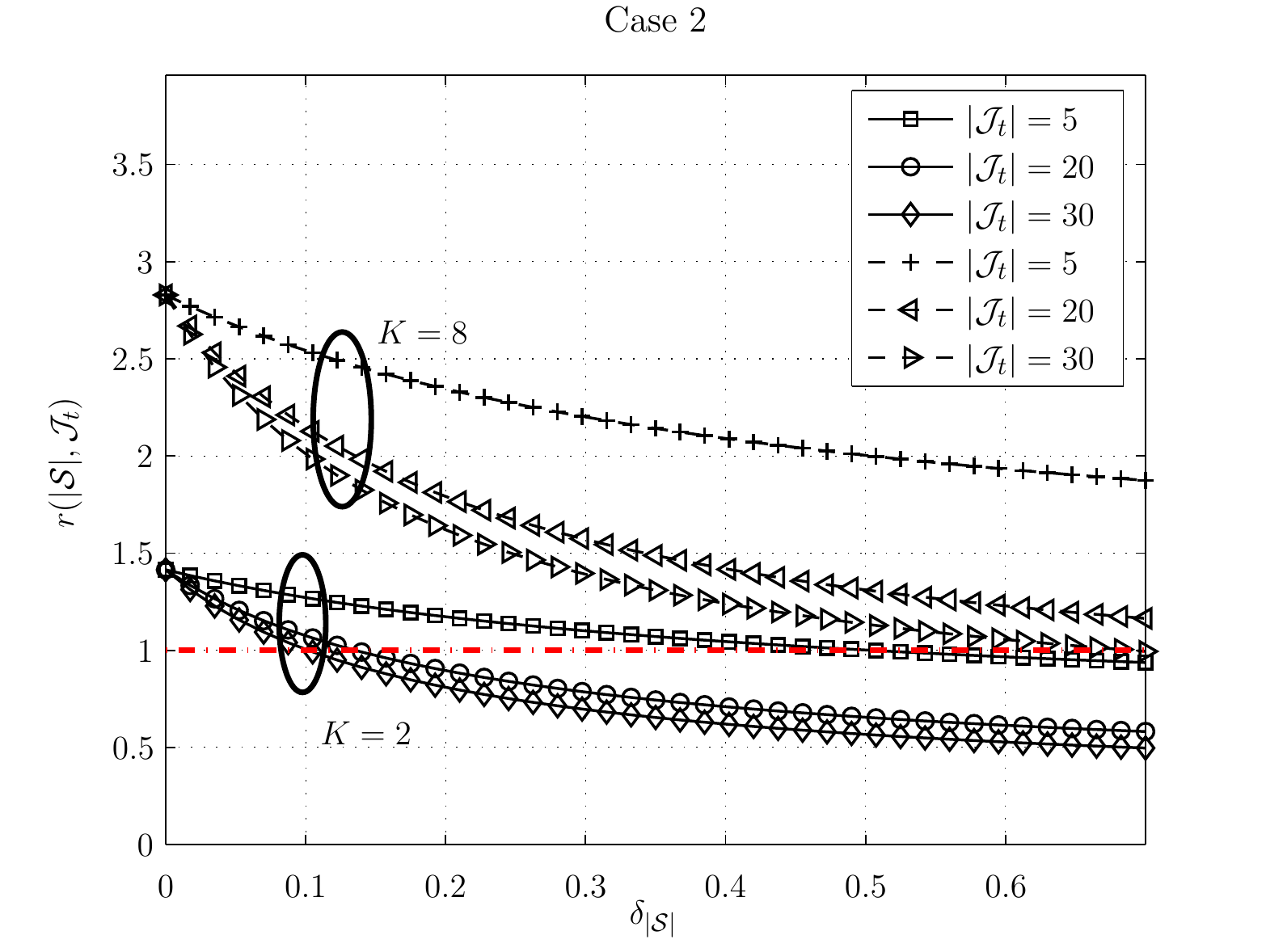}
	\caption{Analysis of $r(|\mathcal{S}|, \mathcal{J}_t)$ in Case 2  for various $2$-tuples $(|\mathcal{J}_t|, \delta_{\mathcal{S}})$ -- The dot-dash horizontal red line determines when both theorems provide equivalent bounds, \textit{i.e.}, $r(|\mathcal{S}|, \mathcal{J}_t) = 1$.}%
	\label{fig:case2LBComparison}%
\end{figure}

\subsection{Case 3: Last iteration}
We assume that only one correct atom has yet to be picked, \textit{i.e.}, $\mathcal{J}_t = \lbrace j_f \rbrace$. Thus,
\begin{equation}
	r(|\mathcal{S}|, \mathcal{J}_t) =   \dfrac{\| \bsy{X}^{\lbrace j_f \rbrace} \|_{1}}{\| \bsy{X}^{\lbrace j_f \rbrace} \|_{2}} \geq 1.
\end{equation}

If the row $\bsy{X}^{\lbrace j_f \rbrace}$ is $1$-sparse, then $r(|\mathcal{S}|, \mathcal{J}_t) = 1$ while $r(|\mathcal{S}|, \mathcal{J}_t) = \sqrt{K}$ whenever its entries have identical absolute values. Both theorems deliver the same performance within the SMV framework but differ as soon as $K > 1$.

\subsection{Case 4: ``Perfect'' measurement matrix}
Let us assume that $\delta_{|\mathcal{S}|} = 0$. Hence,
\begin{equation}
	r(|\mathcal{S}|, \mathcal{J}_t) = \sqrt{|\mathcal{J}_t|}  \dfrac{\| \bsy{X}^{\mathcal{J}_t} \|_{\infty \rightarrow \infty}}{\| \bsy{X}^{\mathcal{J}_t} \|_{\mathrm{F}}} \geq 1.
\end{equation}
Our contribution is at least equivalent to the state-of-the-art bound in this case. As stated in Remark~\ref{rem:newLBSharp}, Theorem~\ref{thm:newLBCorrelationResidual} is sharp for  $\delta_{|\mathcal{S}|} = 0$, \textit{i.e.}, $\|\bsy{\Phi}_{\mathcal{S}}^{\mathrm{T}} \bsy{R}^{(t)} \|_{\infty \rightarrow \infty}  =  \| \bsy{X}^{\mathcal{J}_t} \|_{\infty \rightarrow \infty}$.
\begin{rem}\label{rem:newLBSharp}
	$\delta_{|\mathcal{S}|} = 0 \Rightarrow \|\bsy{\Phi}_{\mathcal{S}}^{\mathrm{T}} \bsy{R}^{(t)} \|_{\infty \rightarrow \infty}  =  \| \bsy{X}^{\mathcal{J}_t} \|_{\infty \rightarrow \infty}$.
\end{rem}
\begin{proof} For each $j_1$, $j_2 \in \mathcal{S}$ such that $j_1 \neq j_2$, we have $|\langle \bsy{\phi}_{j_1}, \bsy{\phi}_{j_2} \rangle | \leq \delta_2 \|\bsy{\phi}_{j_1}\|_2 \|\bsy{\phi}_{j_2}\|_2$  (see \cite[Lemma 2.1]{candes2008restricted}) where $\delta_2 \leq \delta_{|\mathcal{S}|} = 0$ because the RIC is monotonically increasing, \textit{i.e.}, $\delta_s \leq \delta_{s+1}$. Hence, all the atoms comprised within $\bsy{\Phi}_{\mathcal{S}}$ are orthogonal to each other. Similarly to the proof of Theorem~\ref{thm:newLBCorrelationResidual}, we have $\bsy{\Phi}_{\mathcal{S}}^{\mathrm{T}} \bsy{R}^{(t)} =
	\bsy{\Phi}_{\mathcal{S}}^{\mathrm{T}} (\bsy{I} - \bsy{P}^{(t)}) \bsy{\Phi}_{\mathcal{S}} \bsy{X}^{\mathcal{S}} = \bsy{\Phi}_{\mathcal{S}}^{\mathrm{T}}  \bsy{\Phi}_{\mathcal{J}_t} \bsy{X}^{\mathcal{J}_t}$ and $\| \bsy{\Phi}_{\mathcal{S}}^{\mathrm{T}} \bsy{R}^{(t)} \|_{\infty \rightarrow \infty} = \| \bsy{\Phi}_{\mathcal{J}_t}^{\mathrm{T}}  \bsy{\Phi}_{\mathcal{J}_t} \bsy{X}^{\mathcal{J}_t} \|_{\infty \rightarrow \infty}$. The vanishing of the orthogonal projection matrix stems from the orthogonality of the atoms indexed by $\mathcal{S}$. The matrix $(\bsy{I} - \bsy{P}^{(t)})$ indeed projects onto $\mathcal{R}(\boldsymbol{\Phi}_{\mathcal{S}_{t}})^{\perp}$ and  $\mathcal{R}(\boldsymbol{\Phi}_{\mathcal{J}_{t}}) \subset \mathcal{R}(\boldsymbol{\Phi}_{\mathcal{S}_{t}})^{\perp}$ since $\mathcal{R}(\boldsymbol{\Phi}_{\mathcal{S}_{t}}) \perp \mathcal{R}(\boldsymbol{\Phi}_{\mathcal{J}_{t}})$. As $\delta_{|\mathcal{S}|} = 0$, we have \cite[Equation 6.2]{foucart2013mathematical} $\| \bsy{\Phi}_{\mathcal{J}_t}^{\mathrm{T}}  \bsy{\Phi}_{\mathcal{J}_t} - \bsy{I} \|_{2 \rightarrow 2} = 0$ so that $\bsy{\Phi}_{\mathcal{J}_t}^{\mathrm{T}}  \bsy{\Phi}_{\mathcal{J}_t} = \bsy{I}$.
\end{proof}

\section*{Acknowledgments}

The authors would like to thank the Belgian ``Fonds de la recherche scientifique'' for having funded this research.

\newpage
\nocite{*}
\bibliographystyle{abbrv}
\bibliography{mybib}

\end{document}